\documentclass{amsart}
\allowdisplaybreaks[4]
\usepackage{color}

\newcommand{\Z}{\mathbf{Z}}
\newcommand{\R}{\mathbf{R}}

\newcommand{{\ba}}{\bf a}
\newcommand{\ve}{\varepsilon}
\newcommand{\la}{\lambda}
\newcommand{\La}{\Lambda}
\newcommand{\ga}{\gamma}
\newcommand{\pa}{\partial}
\newcommand{\ra}{\rightarrow}
\newcommand{\Om}{\Omega}
\newcommand{\del}{\delta}
\newcommand{\Del}{\Delta}

\newcommand{\na}{\nabla}

\newcommand{\al}{\alpha}

\newcommand{\be}{\begin{equation}}
\newcommand{\ee}{\end{equation}}

\newcommand{\om}{\omega}

\newtheorem{lem}{Lemma}{\bf}{\it}
\newtheorem{remark}{Remark}{\it}{\rm}

\newtheorem{theorem}{Theorem}
\newtheorem{proposition}{Proposition}

\numberwithin{theorem}{section}
\numberwithin{lem}{section}
\numberwithin{equation}{section}
\numberwithin{proposition}{section}
\numberwithin{corollary}{section}
\title[Central Limit Theorem]{A strong central limit theorem for a class of random surfaces}

\author{Joseph G. Conlon and Thomas  Spencer}

\address{ (Joseph G. Conlon): University of Michigan\\ Department of Mathematics\\ Ann Arbor,
  MI 48109-1109}
\email{conlon@umich.edu}

\address{ (Thomas Spencer): School of Mathematics, Institute for Advanced Study, Princeton, NJ 08540}
\email{spencer@math.ias.edu}


\begin{document}

\maketitle

\begin{abstract}
This paper is concerned with $d=2$ dimensional lattice field models with action $V(\na\phi(\cdot))$, where $V:\R^d\ra \R$ is a uniformly convex function. The fluctuations of the variable $\phi(0)-\phi(x)$ are studied for large $|x|$   via the generating function given by
$g(x,\mu) = \ln \langle e^{\mu(\phi(0) - \phi(x))}\rangle_{A}$.
In two dimensions $g''(x,\mu)=\pa^2g(x,\mu)/\pa\mu^2$ is proportional to $\ln\vert x\vert$.  
The main result of this paper is a  bound on $g'''(x,\mu)=\pa^3 g(x,\mu)/\pa \mu^3$
which is uniform in $\vert x \vert$ for a class of convex $V$. 
 The proof uses integration by parts following Helffer-Sj\"{o}strand and Witten, and relies on estimates of singular integral operators on weighted Hilbert spaces.

\end{abstract}

\section{Introduction.}
We shall be interested in probability spaces  $(\Om,\mathcal{F},P)$ associated with certain Euclidean lattice field theories.  These Euclidean field theories are determined by a potential $V : \R^d \ra \R$ which is a $C^2$ uniformly convex function.  Thus the second derivative $V''(\cdot)$ of $V(\cdot)$ is assumed to satisfy the quadratic form inequality  
\begin{equation} \label{A1}
\la I_d \le V''(z) \le \La I_d, \ \ \ \ \ \ z=(z_1,...,z_d)\in \R^d,
\end{equation}
where $I_d$ is the identity matrix in $d$ dimensions and $\la, \La$
are positive constants. The 
measure $P$ is formally given as
\begin{equation} \label{B1}
P =\exp \left[ - \sum_{x\in \Z^d} V\Big( \na\phi(x)\Big) +\frac{1}{2} m^2 \sum_x \phi(x)^2\right] \prod_{x\in \Z^d} d\phi(x)/{\rm normalization},
\end{equation}
where  $m > 0$ and  $\nabla$ is the discrete gradient operator acting on fields $\phi : \Z^d \ra \R$.   In the case when
$V(z) \ = \ |z|^2/2+ a \sum_{j=1}^d\cos z_j, \quad z\in\R^d$, the probability measure (\ref{B1}) describes the dual representation of a gas of lattice dipoles with activity $a$ (see  \cite{brydges}).  Our estimates on fluctuations will be uniform for $m >0$.

We denote the adjoint of $\nabla$ by  $\nabla^*$. Thus $\nabla$ is a $d$ dimensional { \it column} operator and $\nabla^*$ a $d$ dimensional {\it row} operator, which act on  functions $\phi:\Z^d\ra\R$ by
\begin{eqnarray} \label{C1}
\na \phi(x) &=& \big( \na_1 \phi(x),... \ \na_d\phi(x) \big), \quad  \na_i \phi(x) = \phi (x + {\bf e}_i) - \phi(x),  \\
\na^* \phi(x) &=& \big( \na^*_1 \phi(x),... \ \na^*_d\phi(x) \big), \quad  \na^*_i \phi(x) = \phi (x - {\bf e}_i) - \phi(x). \nonumber
\end{eqnarray}
In (\ref{C1}) the vector  ${\bf e}_i \in \Z^d$ has 1 as the ith coordinate and 0 for the other coordinates, $1\le i \le  d$. Note that the Hessian of our action in (\ref{B1}) is a uniformly elliptic finite difference operator acting on $\ell^2(\Z^d)$ given by $$\nabla^*V''(\nabla\phi(\cdot))\nabla+m^2 \;.$$ 

Let $\Om$ be the space of all functions $\phi : \Z^d \ra \R$ and $\mathcal{F}$ be the Borel algebra generated by finite dimensional rectangles 
$\{ \phi \in \Om: \  |\phi(x_i) - a_i| < r_i, \ i=1,...,N\}$, 
$x_i \in \Z^d, \ a_i \in \R, \ r_i > 0, \ i=1,...,N, \ N \ge 1$.   The $d$ dimensional integer lattice $\Z^d$ acts on $\Om$ by translation operators $\tau_x:\Om\ra\Om, \ x\in\Z^d$, where $\tau_x\phi(z)=\phi(x+z), \ z\in\Z^d$. Translation operators  are measurable and satisfy the properties $\tau_x\tau_y=\tau_{x+y}, \ \tau_0= \ {\rm identity}, \ x,y\in\Z^d$.  It was first shown by Funaki and Spohn \cite{fs} that as $m\ra 0$ one can define  a unique ergodic translation invariant probability 
measure $P$ on $(\Om, \mathcal{F})$ corresponding to (\ref{B1}).  If $d\ge 3$ this is a measure on fields $\phi:\Z^d\ra\R$, but for $d=1,2,$ one needs to regard (\ref{B1}) as a measure on gradient fields $\om=\na \phi$.  In that case the Borel algebra $\mathcal{F}$  is generated by finite dimensional rectangles for $\om(\cdot)$ with the usual gradient constraint that the sum of $\om(\cdot)$ over plaquettes is zero.

Estimates on expectation values $\langle\cdot\rangle_\Om$ for $(\Om,\mathcal{F},P)$ can be obtained from the Brascamp-Lieb inequality  \cite{bl}. Since by (\ref{A1}) we have a uniform lower bound  on the Hessian,  this inequality implies that for $f:\Z^d \ra \R$, with $\sum_{y\in\Z^d} f(y) = 0$
 \be \label{D1}
\langle  \exp[(f,\phi)- \langle(f,\phi)\rangle_{\Om} ]\rangle_{\Om} \ \  \le  \ \  \exp\left[\frac{1}{2} (f, (-\la\Del)^{-1} f)\right],
\ee 
where $(\cdot,\cdot)$ denotes the standard inner product for functions on $\Z^d$ and $\Delta$ is the discrete Laplacian on $\Z^d$. Thus (\ref{D1}) bounds all moments of $(f,\phi) - \langle(f,\phi)\rangle$ in terms of $ (f, (-\la\Del)^{-1} f)$.

It follows from (\ref{D1}) that the function  $g(\cdot, \cdot)$ defined by
\be \label{F1}
g(x,\mu)= \log\langle e^{\mu(\phi(0)-\phi(x))}\rangle_\Om, \quad \mu \in \R
\ee
satisfies the inequality $g(x,\mu)\le C_d\mu^2$ for some constant $C_d$ provided $d\ge 3$. If $d=1,2$ then  (\ref{D1}) implies that $g(x,\mu)\le C_d(x)\mu^2$ where $C_2(x)\sim \log |x|$ and $C_1(x)\sim |x|$ for large $|x|$. Since in dimension $d=1$ the random variables $\na\phi(x), \ x\in\Z$, are i.i.d., it is easy to see that in this case $g(x,\mu)=C(\mu)|x|$ for a positive constant $C(\mu)$ depending only on $\mu$. 
In this paper we shall show that the $x$ dependence of $C_2(x)$  for large  $|x|$  is entirely due to the second moment of $\phi(x)-\phi(0)$. 

\begin{theorem} 
Suppose $d=2$ and $V:\R^d\ra\R$ is $C^3$, satisfies the inequality (\ref{A1}) and $\|V'''(\cdot)\|_\infty=M<\infty$.  If in addition $\la/\La>1/2$, then  there is a positive constant $C$ depending only on $\la,\La,$  such that
\be \label{M1}
\big|g'''(x,\mu) = \frac{\pa^3 g(x,\mu)}{\pa\mu^3} |  \le \ C M \quad  x\in\Z^d, \  \mu\in\R.
\ee
Hence we have  $| \ g(x,\mu)-\frac{\mu^2}{2}\langle (\phi(0)-\phi(x))^2\rangle_\Om \ | \ \le \ C \mu^3 M/6, \quad  x\in\Z^d$. The constants above are uniform for $m >0$.

\end{theorem}

\smallskip
\noindent\textbf{Remark:} If $(\phi(0)-\phi(x))$  is Gaussian then $g'''(x,\mu) =0$. Note that in one dimension,   
$g''' (x,\mu)\propto |x|$ unless our distribution is Gaussian. Thus the analog of our theorem is \textit{not} valid in one dimension. 
In this sense, the long range correlation of the gradient fields in 2D give a stronger CLT.

\smallskip
The proof of Theorem 1.1 follows from an inequality for the third moment of $\phi(0)-\phi(x)$, 
\begin{equation} \label{N1}
 \frac{\pa^3 g(x,\mu)}{\pa\mu^3}  =   \langle \ [X-\langle X\rangle_{\Om,x,\mu}]^3 \ \rangle_{\Om,x,\mu} \ , \quad  {\rm where \ }  X=\phi(0)-\phi(x), 
\end{equation}
 and $\langle\cdot\rangle_{\Om,x,\mu}$ denotes expectation with respect to the probability measure proportional to
\be \label{O1}
e^{\mu(\phi(0)-\phi(x))} \ dP(\phi(\cdot)) \ ,
\ee
with $P$ the translation invariant measure (\ref{B1}). In $d\ge 3$ dimensions, (\ref{N1}) is uniformly bounded by applying (1.4)
to $\langle \cdot \rangle_{\Om, \mu, x}.$

If $\mu=0$ and the function $V(\cdot)$ of (\ref{B1}) is symmetric i.e. $V(z)=V(-z), \ z\in\R^d,$  then it is easy to see that the third moment of $\phi(0)-\phi(x)$ is $0$. More generally we have the following decay estimate:
\begin{theorem} Under the assumptions of  Theorem 1.1 we have:
\begin {equation}
\left| \ \langle (\phi(0)-\phi(x))^3\rangle_\Om \ \right| \ \le  \ C M/[1+|x|^\al], \quad x\in\Z^d \ ,
\end{equation}
for some positive $\alpha$.
\end{theorem}
\smallskip 
\noindent\textbf{Relation to dimers:} In two dimensions one can think of $\phi(x), \ x\in\Z^2,$ as being the height of a random surface over  $\Z^2$ which fluctuates logarithmically. Theorem 1.1 was motivated by related results for dimer models.
The uniform measure on dimer covers of the square lattice has an associated height function $\phi(\cdot)$ which takes \textit{integer} values. Denoting by $\langle\cdot\rangle_D$ the expectation on heights induced by the uniform measure on dimers, the height fluctuations $\langle \ (\phi(0)-\phi(x))^2 \ \rangle_D$ grow logarithmically with $|x|$ (see \cite{kenyon1, kenyon2} for an introduction to dimers and heights).  As in Euclidean field theory with measure (\ref{B1}), one can consider the function $g(x,\mu)$ defined by
\be \label{Q1}
g(x,\mu)= \log\langle e^{\mu(\phi(0)-\phi(x))}\rangle_D \ , \qquad x\in\Z^2, 
\ee
but in this case it is interesting to let  $\mu$  be {\it pure imaginary}, whereas in Theorem 1.1 $\mu$ is {\it real}. In \cite{pinson} it is shown that there exists $\del>0$ such that
\be \label{R1}
 | \ g(x,\mu)-\frac{\mu^2}{2}\langle (\phi(0)-\phi(x))^2\rangle_D \ | \ \le \ C, \quad  x\in\Z^2, \  \mu\in i\R, \ |\mu|<\del,
\ee
for some constant $C$.  This implies that $\langle e^{\mu(\phi(0)-\phi(x))}\rangle_D $ has a power law decay which is determined only by the variance.  Since one also has \cite{kenyon1, kenyon2} that
\be \label{S1}
\langle (\phi(0)-\phi(x))^2\rangle_D \  = \  \frac{16}{\pi^2}\log |x|+O(1) \quad {\rm as \ } |x|\ra\infty,
\ee
the inequality (\ref{R1}) gives rather precise information on the behavior of $\langle e^{\mu(\phi(0)-\phi(x))}\rangle_D $ for large $|x|$ and small $\mu$.
   The inequality (\ref{R1}) for $x$ lying along lattice lines follows from earlier work \cite{bw} on Toeplitz determinants for piecewise smooth symbols. These results allow for a larger range of $\del$ in (\ref{R1}) than \cite{pinson} does. Recent work by Deift, Its and Krasovsky \cite{deift} gives optimal estimates for $\mu = i\pi/2 $. In special cases this power law decay is related to the spin-spin correlation of an Ising antiferromagent on a triangular lattice at 0 temperature.  Note that since the heights are integer valued, when $\mu =2\pi i$, $g(x, \mu) \equiv 0$ so that (\ref{R1}) cannot hold for all $\mu$.
   
A closely related central limit theorem arises in fluctuations of the number of eigenvalues of a $U(N)$ matrix belonging to an arc on the circle. The variance of this number grows logarithmically in $N$. If we call the corresponding generating function $g(N, \mu)$, then for a suitable range of $\mu$ we have 
$|g'''(N, \mu)| \le {\rm Constant}$. If the indicator function of the arc is smoothed out, the logarithmic growth in $N$ disappears. Note that the methods of this paper do not apply to dimer heights or to $U(N)$  because the integer constraints make the associated action non-convex. 

\medskip
\noindent\textbf{Idea of Proof:} The  reason that a stronger form of CLT holds in dimension 2 may be understood as follows: We can express 
\be \label{U1}
\phi(0)-\phi(x) = \sum_{y\in\Z^d} \nabla \phi(y)\cdot[ \nabla G_0 (y) - \nabla G_0 (y-x)] \ ,
\ee
where $G_0$ is the Green's function of the discrete 
Laplacian. In 2D the sum of the gradients is spread out since $|\nabla G_0 (y)| \sim (|y| + 1)^{-1}$. Although this function is not
square summable, it lies in  the weighted $\ell^2$ space  $\ell^2_w(\Z^2,\R^2)$ with weight $w(y)=[1+|y|]^\al$ for any $\al<0$.   Hence from the theory of singular integral operators \cite{stein2} the convolution of $\na\na^*G_0(\cdot)$ with $\na G_0(\cdot)$ is also in the space  $\ell^2_w(\Z^2,\R^2)$. This situation should be contrasted with the case of one dimension where the gradient has no decay.  

In order to implement our argument, which is based on the intuition gained from (\ref{U1}) and the decay of the $2D$ Green's function, we use an integration by parts formula due to Helffer-Sj\"{o}strand  and Witten \cite{helffer,hs}, and some results on singular integral operators on weighted spaces.  The integration by parts formula can be stated formally as
\be \label{T1}
\langle (F_1- \langle F_1\rangle ) \,F_2\rangle_{\Om,x,\mu} \ =  \ \langle dF_1\,\cdot [d^*d+\nabla^*V''(\nabla\phi(\cdot))\nabla+ m^2]^{-1}dF_2\rangle_{\Om,x,\mu} \ .
\ee
In (\ref{T1}) expectation is with respect to the measure (\ref{O1}), and the functions $F_i(\phi(\cdot)) $ are differentiable functions of the field $\phi:\Z^d\ra\R$. The operator $d$ is the gradient operator acting on functions of $\phi(\cdot)$, and $d^*=-d + \nabla V + \mu \nabla X$  is the corresponding divergence operator with respect to the measure (\ref{O1}).  The dot product on the right side of (\ref{T1}) is over the lattice sites indexing the gradient. Note that   $d^*d$ is the elliptic self-adjoint operator acting on functions of $\phi(\cdot)$, which corresponds to the Dirichlet form for (\ref{O1}).  The identity (\ref{T1}) is explained in more detail in the following section, and since $d^*d$ is nonnegative  it implies (\ref{D1}). The operator $d^*d+\nabla^*V''(\nabla\phi(\cdot))\nabla $ formally acts  on functions $F(y,\phi(\cdot))$. The first term acts  as a differential operator in the field variable $\phi(\cdot)$,  and the second term acts as a finite difference elliptic operator in the lattice index $y$.

We first prove an $L^2$ version of Theorem 1.1 using the integration by parts formula (\ref{T1}). This result unfortunately requires the seemingly artificial restriction $\la/\La>1/2$ on the bounds (\ref{A1}). 
The reason for the restriction on $\la/\La$ is that we need to express our Green's function so that second order finite difference derivatives $\nabla_x$ appear in a symmetric way. This problem arises due to the presence of the operator $d^*d$, and therefore does not occur in the classical case where we set $d^*d\equiv 0$. See Lemma 2.2 and the resolvent expansion for (\ref{AE2}).

Theorems 1.1 and 1.2  follow  from an extension of the $L^2$ theorem to the corresponding theorem for weighted $L^2$ spaces,  with weights which are in the Muckenhoupt $A_2$ class \cite{stein2}. The weights can be chosen arbitrarily close to the constant function in the $A_2$ norm, and so Theorem 1.1 also holds with the restriction $\la/\La>1/2$.  The reason for  this is that the norm of a Calderon-Zygmund operator on an $A_p$ weighted space  is a continuous function of the $A_p$ norm at the constant. This continuity result does not follow from the standard proofs \cite{stein2} of the boundedness of Calderon-Zygmund operators on weighted spaces, and   was proven quite recently \cite{pvol}.  If one however restricts  to weights which are dilation and rotation invariant, continuity follows from  the argument in a classical paper on the subject \cite{stein1}.  The weights considered in this paper are approximately rotation and dilation invariant. 

In $\S2$  we first state and prove (\ref{T1}) and then obtain an estimate on the third moment of  $\sum_x h(x) \nabla \phi(x)$ with h in $\ell_2(\Z^2)$. There are two proofs for the third moment. The first proof  uses quadratic form inequalities, and the second uses a convergent perturbation expansion. In $\S4$ it is shown that the perturbation expansion also converges for functions in weighted spaces with weight close to $1$. As explained above this is needed to prove Theorem 1.1 since $(\phi(0)- \phi(x))= \sum_x  h(x) \nabla \phi(x)$  with $h$ in a weighted $\ell_2$ space. See (\ref{U1}).  The required weighted norm inequalities for functions on $\Z^d$ are proved in $\S3$. These inequalities are applied to the field theory setting in $\S4$ by using the spectral decomposition of the self-adjoint operator $d^*d$.    Because we need to make use of the spectral decomposition theorem, we cannot replace the weighted norm inequalities in our argument by $L^q$ inequalities with $q$ close to $2$.

\section{The $L^2$ Theory}
Our main goal in this section will be to establish an $L^2$ version of Theorem 1.1. First we shall state and prove a finite dimensional  version of  the Helffer-Sj\"{o}strand formula (\ref{T1}) which we shall use in the proof.

Let $L$ be a positive even integer and $Q=Q_L\subset \Z^d$ be the integer lattice points in the closed  cube centered at the origin with side of length $L$. By a periodic function $\phi:Q\ra\R$ we mean a function $\phi$ on $Q$ with the property that $\phi(x)=\phi(y)$ for all $x,y\in Q$ such that $x-y=L{\bf e}_k$ for some $k, \ 1\le k\le d$. Let $\Om_Q$ be the space of all periodic functions $\phi:Q\ra\R$, whence $\Om_Q$ with $Q=Q_L$ can be identified with $\R^{N}$ where $N=L^d$. Let $\mathcal{F}_Q$ be the Borel algebra for $\Om_Q$ which is generated by the open sets of $\R^{N}$. For $m>0$, we define a probability measure $P_{Q,m}$ on $(\Om_Q,\mathcal{F}_Q)$ as follows:
\begin{multline} \label{B2}
\langle F\rangle_{\Om_Q,m}= \\
 \int_{\R^N} F(\phi)\exp\left[-\sum_{x\in Q} \left\{V(\nabla \phi(x))+\frac{1}{2}m^2\phi(x)^2\right\}\right]  \prod_{x\in Q} d\phi(x)/{\rm normalization} \ ,
 \end{multline}
 where $F:\R^N\ra\R$ is a continuous function such that $|F(z)|\le C\exp[A|z|], \ z\in\R^N$, for some constants $C,A$.  Note that   $\langle \ \phi(x) \ \rangle_{\Om_Q,m}= 0$ for all $x\in Q$. This follows from the translation invariance of the measure  (\ref{B2}), upon  making the change of variable $\phi(\cdot)\ra\phi(\cdot)+\ve$, differentiating with respect to $\ve$ and setting $\ve=0$. 
We consider now for $\mu\in\R$ and $x\in Q$  the probability measure   proportional to the measure  
\be \label{I2}
e^{\mu(\phi(0)-\phi(x))} \ dP_{Q,m}(\phi) \ 
\ee
 on $(\Om_Q,\mathcal{F}_Q,P_{Q,m})$, which is  analogous to (\ref{O1}), and  denote expectation with respect to this measure by $\langle\cdot\rangle_{\Om_Q,m,x,\mu}$.  Let  $F:\R^N\ra\R$ be a $C^1$ function and  $dF:\R^N\ra\R^N$ be its gradient.  For a $C^1$ function $G:\R^N\ra\R^N$ the divergence $d^*G$  of $G$ with respect to the measure (\ref{I2}) is formally defined from the integration by parts formula
\be \label{J2}
\langle (G,dF)\rangle_{\Om_Q,m,x,\mu} \ \ = \ \   \langle(d^*G,F)\rangle_{\Om_Q,m,x,\mu} \ .
\ee
\begin{lem}[Helffer-Sj\"{o}strand]
Let $F_1,F_2$ be two $C^1$ functions on $\R^N$ such that for $j=1,2,$ the inequality $|F_j(z)|+|DF_j(z)|\le C\exp[A|z|], \ z\in\R^N,$ holds for some constants $C,A$. If $\langle F_2\rangle_{\Om_Q,m,x,\mu}=0$  then there is the identity
\be \label{K2}
\langle F_1F_2\rangle_{\Om_Q,m,x,\mu} \ =  \ \langle dF_1[d^*d+\nabla^*V''(\nabla\phi(\cdot))\nabla+m^2]^{-1}dF_2\rangle_{\Om_Q,m,x,\mu} \ .
\ee
 Note that $d^*$ is the adjoint with respect to the measure given by (2.2).
\end{lem}
 Sketch of proof: Since $d^*d$ generates a compact semigroup on a bounded domain with a unique groundstate $1$ and  $\langle F_2\rangle = 0$, it follows  that there exists a solution $F_3$  to the equation 
\be \label{AM2}
d^*dF_3=F_2 \ , \quad {\rm implies \ } \langle  \   F_1 F_2 \ \rangle \ = \ 
\langle  \   dF_1 \,dF_3 \ \rangle \ .
\ee 
If we assume that $F_2(\cdot)$ is a $C^{1+\al}$ function for any $\al>0$ then elliptic regularity theory implies that $F_3(\cdot)$ is $C^3$. The identity (\ref{K2}) follows  from (\ref{AM2}) by observing that
\be \label{AN2}
dF_2 \ = \ (dd^*)\,dF_3 \ = \ [\,d^*d+\nabla^*V''(\nabla\phi(\cdot)\nabla+m^2] dF_3 \ .
\ee 
Note that above we have used the fact that the commutator  $[d^*,d] $ is  the Hessian. For details see \cite{fs}.

 \begin{theorem} Suppose $d\ge 1$ and the constants in (\ref{A1}) satisfy $\la/\La>1/2$. Then there is a positive constant $C(\la,\La)$ depending only on $\la,\La$  such that for any $h_1,h_2,h_3\in \ell^2(\Z^d,\R^d)$ and  $ x\in\Z^d, \  \mu\in\R$,
\be \label{A2}
|\langle  \ \prod_{j=1}^3\big[ (h_j,\na\phi)-\langle (h_j,\na\phi) \rangle_{\Om,x,\mu}\big] \ \rangle_{\Om,x,\mu}| \ \le \ C(\la,\La)\|h_1\|\|h_2\|\|h_3\| \sup_{\xi\in \R^d} |V'''(\xi)| \ .
\ee
\end{theorem}
The proof of Theorem 2.1 depends on a representation for the third moment of  $ (h,\na\phi)$, which we obtain by applying the Helffer-Sj\"{o}strand formula of Lemma 2.1. We first obtain the representation for a periodic cube in $\Z^d$ and  then conclude from \cite{fs} that the representation continues to be valid as the cube increases to $\Z^d$. Note that this theorem is not sufficient to imply our main estimate (\ref{M1}) because, if we express $\phi(0) - \phi(x)$ in terms of $\na\phi \cdot h$ as in (\ref{U1}), the $\ell_2$ norm of $h$  diverges logarithmically for large $|x|$. In sections 3 and 4 we will show how to use weighted norms to solve this problem.

 Let $h_j:Q\ra\R^d, \ j=1,2,3$ be arbitrary periodic functions and define $G_j(\phi(\cdot))$ in terms of them by
\be \label{L2}
G_j    \ = \ \sum_y\left[h_j(y)\cdot\na\phi(y)-\langle  h_j(y)\cdot\na\phi(y)  \rangle_{\Om_Q,m,x,\mu}\right]  \ . 
\ee
Applying (\ref{K2}) to the functions $F_1=G_1  G_2$ and $ F_2=G_3$  yields the identity
\begin{multline} \label{M2}
\langle  G_1 G_2 G_3    \rangle_{\Om_Q,m,x,\mu} \ =  \\
\langle  \big( \ \left[ \ G_1\na^*h_2(\cdot)+ G_2\na^*h_1(\cdot) \ \right], \Phi_3(\cdot,\phi) \ \big)    \rangle_{\Om_Q,m,x,\mu} \ ,
\end{multline}
where $\Phi_j(y,\phi), \ y\in Q, \  \phi(\cdot)\in\Om_Q, \ j=1,2,3$ is the solution to the equation
\be \label{N2}
\left[d^*d+\nabla^*V''(\nabla\phi(y))\nabla+m^2\right] \Phi_j(y,\phi(\cdot)) \ =   \  dG_j \ = \  \na^*h_j(y), \quad y\in Q.
\ee
Since for each $y\in Q$ the expectation $\langle  \ \left[G_1\na^*h_2(y)+ G_2\na^*h_1(y)\right] \  \rangle_{\Om_Q,m,x,\mu} =0$, we can apply (\ref{K2}) again to the RHS of (\ref{M2}). Thus we obtain the identity
\begin{multline} \label{O2}
\langle  G_1 G_2 G_3    \rangle_{\Om_Q,m,x,\mu} \ =  \\
\sum_{y,z\in Q} \langle   \ \left[\Phi_1(z,\phi)\na^*h_2(y)+ \Phi_2(z,\phi)\na^*h_1(y)\right] \  d\Phi_3(y,z,\phi) \     \rangle_{\Om_Q,m,x,\mu} \ ,
\end{multline}
where $d\Phi_j(y,z,\phi), \ z\in Q, $ is the gradient of the function $ \Phi_j(y,\phi)$ which is the solution to (\ref{N2}). Since $ \Phi_j(\cdot,\phi)$ itself is the gradient of a function of $\phi(\cdot)$  it follows that  $d\Phi_j(y,z,\phi) $  is symmetric in $(y,z)$. By applying $d$ to (\ref{N2}) and noting that $\pa/\pa \phi(z) V''(\nabla \phi(y)) = V'''(\nabla \phi(y))\nabla \delta (y-z)$, it is easy to see that $d\Phi_ 3(y,z,\phi) $ is the solution to the equation
\begin{multline} \label{P2}
\sum_{y,z\in Q}f_1(y)f_2(z)\left[d^*d+\nabla_y^*V''(\nabla\phi(y))\nabla_y+\nabla_z^*V''(\nabla\phi(z))\nabla_z+2m^2\right] d\Phi_ 3(y,z,\phi) \\
 = \ -\sum_{y,z\in Q}  V'''(\na\phi(y))[\na f_1(y), \na f_2(z),\na\Phi_3(y,\phi)]\delta(y-z), \quad f_1,f_2:Q\ra \R.
\end{multline}
Here $V'''(\xi)[\cdot,\cdot,\cdot]$ denotes the symmetric trilinear form which is the third derivative of $V(\xi), \ \xi\in\R^d$ and. Let $\Psi(y,z,\phi) $ be the solution to the equation
\begin{multline} \label{Q2}
L \Psi \equiv \left[d^*d+\nabla_y^*V''(\nabla\phi(y))\nabla_y+\nabla_z^*V''(\nabla\phi(z))\nabla_z+2m^2\right] \Psi(y,z,\phi) \\
 = \left[\Phi_1(z,\phi)\na_y^*h_2(y)+ \Phi_2(z,\phi)\na_y^*h_1(y)\right] \quad y,z\in Q  .
 \end{multline}
 It follows from (\ref{O2}), (\ref{P2}), (\ref{Q2}) that
  \begin{multline} \label{AH2}
\langle  G_1 G_2 G_3    \rangle_{\Om_Q,m,x,\mu}\ \ =  \\
-\sum_{y,z\in Q} \langle \  V'''(\na(\phi(y))[ \na_y\na_z\Psi(y,z,\phi), \   \na_y\Phi_3(y,\phi)] \ \rangle_{\Om_Q,m,x,\mu} \delta(y-z) \ .
\end{multline}
 Above we have also used  the fact that $L^{-1}$ is a symmetric operator so that it may be transfered to the first factor of (\ref{O2}). Since $\delta$ on the right side of (\ref{AH2}) is just the Kronecker delta function we can apply the Schwarz inequality to obtain
 \begin{multline} \label{R2}
|\langle  G_1 G_2 G_3   \rangle_{\Om_Q,m,x,\mu}| \ \le  \\
\sup_{\xi\in\R^d}|V'''(\xi)|\left\{\sum_{y,z\in Q} \langle \   |\na_y\na_z\Psi(y,z,\phi)|^2  \ \rangle_{\Om_Q,m,x,\mu}\right\}^{1/2} \ \left\{\sum_{y\in Q} \langle \   | \na_y\Phi_3(y,\phi)|^2 \ \rangle_{\Om_Q,m,x,\mu}\right\}^{1/2} \ .
\end{multline}
From (\ref{N2}) the second term in curly braces on the RHS of (\ref{R2}) is bounded by $\|h_3\|/\la$.
This follows from the fact that the norm of the operator 
\be \label{AO2}
\nabla_y\left[d^*d+\nabla_y^*V''(\nabla\phi(y))\nabla_y+m^2\right]^{-1}\nabla^*_y
\ee
is bounded uniformly for $m>0$. Note that $\nabla$ appears symmetrically in this expression so that quadratic form bounds apply.
To estimate the third moment in terms of the $L^2$ norms of the $h_j(\cdot), \ j=1,2,3$ we need to bound the first term in curly braces. 
To do this we must rearrange the gradients so that they are in symmetric form. This is done below by expanding in a Neumann series which enables us to shuffle the lattice gradients so they appear in symmetric form. The convergence of this series is where we need the condition on $\la/\La$.
\begin{lem} Let $\Psi$ be given by (\ref{Q2}) and set 
\be \label{AP2}
\Phi = \left[\Phi_1(z,\phi)h_2(y)+ \Phi_2(z,\phi)h_1(y)\right].
\ee
Then denoting expectation on $\Om_Q$ by $\langle\cdot\rangle$, there is a constant $C(\la,\La)$ depending only on the constants in (\ref{A1}) such that
\be \label{Y2}
\sum_{y,z\in Q} \langle \   |\na_y\na_z\Psi(y,z,\phi)|^2  \ \rangle \ \le  \ 
C(\la,\La)\sum_{y,z\in Q} \langle \   |\nabla_z\Phi(y,z,\phi)|^2  \ \rangle  \ ,
\ee
provided $\la/\La>1/2$.
\end{lem}
Using this lemma the proof of Theorem 2.1 is easy.
\begin{proof}[Proof of Theorem 2.1]
From (\ref{N2}) it follows that
\be \label{AA2}
\sum_{y,z\in Q} \langle \   |\nabla_z\Phi(y,z,\phi)|^2 \ \rangle \  \le \ [2\|h_1\|\|h_2\|/\la]^2 \ .
\ee
Now Lemma 2.2 and (\ref{R2}) imply that
\be \label{AB2}
|\langle  G_1 G_2 G_3    \rangle_{\Om_Q,m,x,\mu}| \ \le 
\sup_{\xi\in\R^d}|V'''(\xi)| C(\la,\La) \|h_1\|\|h_2\|\|h_3\|
\ee
provided $\la/\La>1/2$. The result follows from (\ref{AB2}) and \cite{fs} on letting $Q\ra\Z^d$ and $m\ra 0$.
\end{proof}

\begin{proof}[Proof of Lemma 2.2] 
The $\nabla_z$  must be transfered  to $\Phi$. For this reason
we introduce an elliptic system  symmetric with respect to  permutation of $z$ and $y$, which enables us to use standard quadratic form methods to bound $\na_y\na_z\Psi(y,z,\phi(\cdot))$. From (\ref{A1}) we have that $V''(\xi)=\La[I_d-{\bf b}(\xi)], \ \xi\in\R^d$, where ${\bf b}(\cdot)$ satisfies the quadratic form inequality $0\le {\bf b}(\cdot)\le (1-\la/\La)I_d$. We consider the system
\begin{multline} \label{V2}
\Big\{\left[d^*d+\nabla_y^*V''(\nabla\phi(y))\nabla_y+\La\na^*_z\na_z +2m^2\right]\Psi_1(y,z,\phi) 
\\
 -\La\nabla_y^*{\bf b}(\nabla\phi(y))\nabla_y \Psi_2(y,z,\phi) \Big\}
 = \nabla^*_y\Phi(y,z,\phi), \quad y,z\in Q,
 \end{multline}
\begin{multline} \nonumber
\Big\{\left[d^*d+\nabla_z^*V''(\nabla\phi(z))\nabla_z+\La\na^*_y\na_y +2m^2\right]\Psi_2(y,z,\phi) 
\\
 -\La\nabla_z^*{\bf b}(\nabla\phi(z))\nabla_z \Psi_1(y,z,\phi)\Big\} 
 = 0, \quad y,z\in Q.
 \end{multline}
 Note that in the first equation of (\ref{V2}) the operator $\na_z$ commutes with the operator in the curly braces, and in the second equation the operator $\na_y$ similarly commutes. 
By adding the two equations we see that $L \Psi=L\Psi_1+L\Psi_2 = \nabla^*_y\Phi$ as in (\ref{Q2}) and $\Psi = \Psi_1 +\Psi_2$. 

We  generate the solution to (\ref{Q2}) by means of a converging perturbation expansion in ${\bf b}(\cdot)$. 
Let $T_1,T_2$ be defined by
\begin{eqnarray} \label{AC2}
T_1 \ & \equiv& \  \na_y\left[d^*d/\La+\nabla_y^*\nabla_y+\na^*_z\na_z +2m^2/\La\right]^{-1}\na^*_y \ , \\
T_2 \ & \equiv& \  \na_z\left[d^*d/\La+\nabla_y^*\nabla_y+\na^*_z\na_z +2m^2/\La\right]^{-1}\na^*_z \  .\nonumber 
\end{eqnarray}
We  define a matrix ${\bf B}(\cdot)$   by
\be \label{AD2}
  {\bf B}(y,z,\phi) \ = \  
 \begin{bmatrix} 
       {\bf b}(\nabla\phi(y))& {\bf b}(\nabla\phi(y)) &\\
     {\bf b}(\nabla\phi(z)) &  {\bf b}(\nabla\phi(z)) & \\
    \end{bmatrix}  \  ,
\ee
where $V''(\xi)=\La[I_d-{\bf b}(\xi)]. $ Then  one can check that (\ref{V2})  is equivalent to the system
\begin{multline} \label{AE2}
 \begin{bmatrix}
     \na_z\na_y \Psi_1(y,z,\phi) & \\
    \na_z\na_y \Psi_2(y,z,\phi) & \\
    \end{bmatrix}
    = \La^{-1}
     \begin{bmatrix}
       T_1& 0 &\\
     0 & T_2  & \\
    \end{bmatrix}
   \begin{bmatrix}
     \na_z\Phi(y,z,\phi)  & \\
     0 & \\
    \end{bmatrix}   
 \\ +
 \begin{bmatrix}
      T_1 & 0 &\\
     0 &  T_2 & \\
    \end{bmatrix}
    {\bf B}(y,z,\phi) 
     \begin{bmatrix}
     \na_z\na_y \Psi_1(y,z,\phi) & \\
    \na_z\na_y \Psi_2(y,z,\phi) & \\
    \end{bmatrix}
    \ .
\end{multline}
It is evident that  the operator ${\bf B}(\cdot)$ of (\ref{AD2}) has $L^2$ norm less than $2(1-\la/\La)$. Since the operators $T_1,T_2$ defined by (\ref{AC2}) have norm less than $1$, the Neumann series for the solution of (\ref{AE2}) converges in $L^2(Q\times Q\times \Om_Q,\R^d\times\R^d)$ if $\la/\La>1/2$.
 Thus solving for $\na_z\na_y \Psi(y,z,\phi)$ completes the proof of the lemma. 
 \end{proof}
The solution of (\ref{N2}) can also be generated by a converging perturbation expansion in ${\bf b}(\cdot)$ in the usual way. Thus let  $\Psi(y,\phi)$ be the solution to the equation
\be \label{AF2}
\left[d^*d/\La+\nabla^*\nabla+m^2/\La\right] \Psi(y,\phi) \ = \  \na^*\Phi(y,\phi), \quad y\in Q.
\ee
Then we write $\na\Psi=T\Phi$ which defines the operator $T$. Equation (\ref{N2}) is equivalent to
\be \label{AG2}
\na \Psi(y,\phi) \ = \  \La^{-1}T \Phi(y,\phi) + T[{\bf b}(\na\phi(y)) \na \Psi(y,\phi)] \ ,
\ee
with $ \Phi(y,\phi)=h_j(y), \ y\in Q$. Since the operator $T$ has norm which does not exceed $1$, the Neumann series for the solution of (\ref{AG2}) converges  for any $\la/\La>0$ in $L^2(Q\times\Om_Q,\R^d)$ with measure (\ref{I2}) on $\Om_Q$.

\vspace{.2in}

\section{Weighted Norm Inequalities on $\ell^2$ Spaces}
In this section we prove the weighted norm inequalities on $\ell^2$ spaces which we shall need to prove Theorem 1.1. This section is independent of 
the previous one.
For a  positive periodic function $w:Q\ra \R$ the associated weighted space $\ell_w^2(Q,\R^d)$ is all periodic functions $h:Q\ra\R^d$ with norm $\|h\|_w$ defined by
\be \label{A3}
\|h\|_w^2 \ = \ \sum_{y\in Q} w(y)|h(y)|^2 \  .
\ee
We shall define weights $w(y), W(y,z)$ which grow or decay very slowly. They satisfy the $l_2$ Muckenhaupt condtion which assures us that a natural class of singular integral operators is bounded. Moreover since the weights are slowly varying the weighted operator norms are close
to the $l_2$ norm by recent results of Pattakos and Volberg \cite{pvol}. This approach is made more precise below. In the next section we shall use these weights
to obtain our main theorem.

Define the Green's function on $\Z^d$ by 
\be \label{A4}
[\na^*\na+\rho]G_\rho(y) \ = \ \del(y), \quad y\in\Z^d \ .
\ee
Thus we have that
\begin{multline}  \label{D4}
|\na G_\rho(y)| \ \le \  C/[1+|y|]^{d-1}, \ \  |\na\na^*G_\rho(y)| \ \le \  C/[1+|y|]^d,  \\
 |\na \na\na^*G_\rho(y)| \ \le \  C/[1+|y|]^{d+1}, \quad y\in\Z^d, \ \rho>0,
\end{multline} 
for some constant $C$ depending only on $d$.
The corresponding periodic Green's function for the cube $Q$ with side of length $L$ is
\be \label{B4}
G_{\rho,Q}(y) \ = \ \sum_{y'\in\Z^d} G_\rho(y+Ly') \ .
\ee
In order to estimate the periodic Green's function we need in addition to (\ref{D4}) the inequalities
\be \label{E4}
\left|\sum_{y'\in\Z^d-\{0\}} \na G_\rho(y+Ly') \ \right| \ \le \ C/L^{d-1},   \ee
\be \left|\sum_{y'\in\Z^d-\{0\}} \na\na^*G_\rho(y+Ly') \ \right| \ \le \ C/L^d,  \quad y\in Q, \ \rho>0,
\ee
which hold for a constant $C$ depending only on $d$.  Note that the sums in (\ref{E4}) are not absolutely convergent uniformly for $\rho>0$.
The  Calderon-Zygmund operator $T_\rho$  in a periodic domain is explicitly given by the formula
\be \label{C4}
T_\rho h(y) \ = \ \sum_{y'\in Q} \na\na^*G_{\rho,Q}(y-y') h(y') \ ,
\ee
where $G_{\rho,Q}(\cdot)$ is the function (\ref{B4}). The inequalities (\ref{D4}), (\ref{E4}) therefore yield an estimate on the kernel of $T_\rho$, which is independent of $\rho>0$. The basic proposition for this section may be stated as follows.
 \begin{proposition}
Let $w:Q\ra\R$ be given by $w(y)=[1+|y|]^\alpha, \ y\in Q$, where   $|\alpha |\le d/2$. Then $T_\rho$ is bounded on $\ell_w^2(Q,\R^d)$  for $\rho>0$, and $\|T_\rho\|_w\le 1+C|\alpha|$ for some constant $C$ depending only on $d$.
\end{proposition}
\begin{proof} Adapting the methods of Chapter V of  \cite{stein2}  to the periodic lattice, it is clear in view of the inequalities (\ref{D4}), (\ref{E4})  that the result holds for $ |\al|=d/2$. Now we apply the real interpolation theorem of  \cite{sw} (Theorem  2  in the recent work of Pattakos and Volberg \cite{pvol}) to obtain the inequality for $\|T_\rho\|_w$ when $|\al|$ is small. 

To apply the interpolation theorem we define interpolation measures $\mu_s, \ 0\le s\le 1$ on subsets of the periodic cube by 
$$ \mu_s(E) \ = \ \sum_{y\in E} [1+|y|]^{s\al d/2|\al|}   \quad {\rm for \ } E\subset Q\cap\Z^d,$$
and denote by $\|T_\rho\|_s$ the norm of $T_\rho$ on the space $\ell^2(Q,\R^d,\mu_s)$.  Then by going to the Fourier representation we see that $\|T_\rho\|_0\le1$. It follows from the argument in Chapter V of \cite{stein2} that there is a constant $K$ depending only on $d$  such that  $\|T_\rho\|_1\le K$.  Now Theorem 2 of \cite{pvol} implies that
$$ \|T_\rho\|_w \ = \ \|T_\rho\|_{2|\al|/d} \ \le \  K^{2|\al|/d}  \ \le \ 1+C|\al| \ ,$$
for a constant $C$ depending only on $d$. 
\end{proof} 
Next we consider operators on weighted function spaces of two variables.  For a  positive periodic function 
$ W:Q\times Q\ra \R$ the associated weighted space $\ell_W^2(Q\times Q,\R^d\times\R^d)$ is all periodic functions $h:Q\times Q \ra\R^d\times \R^d$ with norm $\|h\|_W$ defined by
\be \label{E3}
\|h\|_W^2 \ = \ \sum_{(y,z)\in Q\times Q} W(y,z)|h(y,z)|^2 \  .
\ee
Let $T_\rho\otimes I$ be the operator on $\ell_W^2(Q\times Q,\R^d\times\R^d)$ which acts by the operator $T_\rho$ defined by (\ref{C4}) on the $y$ variable of a function $h(y,z)$ and by the identity on the $z$ variable.  
\begin{proposition}
Let $W:Q\times Q\ra\R$ be given by $W(y,z)=[1+|y|]^\alpha[1+\ga(z,y)]^\beta, \ (y,z)\in Q\times Q$, where  
$\gamma(z,y)$ is the shortest distance from $z$ to $y$ on the periodic cube $Q$. Then if  $|\alpha|, \ |\beta|\le  d/2$ the operator  $T_\rho\otimes I$ is bounded on $\ell_W^2(Q\times Q,\R^d\times \R^d)$  for $\rho>0$, and $\|T_\rho\otimes I\|_W\le 1+C[|\alpha|+|\beta|]$  for some constant $C$ depending only on $d$.
\end{proposition}
\begin{proof}   We argue as in Proposition 3.1. Here the interpolating measures are given by
$$ \mu_s(E) \ = \ \sum_{(y,z)\in E} [1+|y|]^{s\al d/2(|\al|+|\beta|)}[1+\ga(y,z)]^{s\beta d/2(|\al|+\beta|)}   \quad {\rm for \ } E\subset (Q\times Q)\cap\Z^{2d},$$ 
with $s=2(|\al|+|\beta|)/d\le 1$ corresponding to $W$.
\end{proof}
For $\rho>0$ let  $T_{1,\rho}$ be the operator on periodic functions $h:Q\times Q \ra\R^d\times \R^d$  defined by 
\be \label{H3}
 T_{1,\rho} \  \equiv \  \na_y\left[\nabla_y^*\nabla_y+\na^*_z\na_z +\rho\right]^{-1}\na^*_y  \ .
\ee
 Let $G_\rho(y,z), \ y,z\in\Z^d$, be the Green's function for the discrete Laplacian on  the lattice of twice the dimension, $\Z^{2d}$ defined as in (\ref{A4}), and $G_{\rho,Q\times Q}(y,z), \ y,z\in Q$, be the corresponding periodic Green's function for the cube $Q\times Q$ defined as in (\ref{B4}). The operator $T_{1,\rho}$ is explicitly given by the formula
\be \label{U4}
T_{1,\rho} h(y,z) \ = \ \sum_{(y',z')\in Q\times Q} \na_y\na_y^*G_{\rho,Q\times Q}(y-y',z-z') \  h(y',z') \ .
\ee
In (\ref{U4}) the row vector $\na_y^*G_{\rho,Q}(y-y',z-z')$ acts on the $y'$ array of the double array column vector $h(y',z')$. 
\begin{proposition}
Let $W:Q\times Q\ra\R$ be given by $W(y,z)=[1+|y|]^\alpha[1+\ga(z,y)]^\beta$ or $W(y,z)=[1+|z|]^\alpha[1+\ga(z,y)]^\beta, \ (y,z)\in Q\times Q$, where  
 $|\alpha |, |\beta|\le d/2$. Then $T_{1,\rho}$ is bounded on $\ell_W^2(Q\times Q,\R^d\times \R^d)$  for $\rho>0$, and  $\|T_{1,\rho}\|_W\le 1+C[|\alpha|+|\beta|]$  for some constant $C$ depending only on $d$.
\end{proposition}
\begin{proof}  Same as for Proposition 3.2.
\end{proof} 
 \begin{remark} We shall assume in the next section that $\alpha,\beta$ are small.
\end{remark} 

\vspace{.2in}

\section{Weighted $L^2$ Theory}
Our goal in this section will be to extend Theorem 2.1 to allow the functions $h_j:Q\ra\R^d, \ j=1,2,3$, to lie in certain {\it weighted} $\ell^2$ spaces. This is needed for the proof of Theorem 1.1. In order to carry this out we define  weighted versions of the $L^2$ spaces of $\S2$.  Thus for  a periodic weight $w:Q\ra\R$  the  associated weighted space  $L_w^2(Q\times\Om_Q,\R^d)$  is the space of all periodic measurable functions $\Phi:Q\times\Om_Q\ra\R^d$ with finite norm  $\|\Phi\|_w$ given by
\be \label{B3}
\|\Phi\|_w^2 \ = \ \sum_{y\in Q} w(y)\langle \ |\Phi(y,\phi)|^2\  \rangle_{\Om_Q,m,x,\mu} \  .
\ee
Letting $T$ be  the operator defined by (\ref{AF2}), it follows  from the spectral decomposition theorem for $d^*d$,  that  $T$ is bounded on $L_w^2(Q\times\Om_Q,\R^d)$ since the operator $T_\rho$ of (\ref{C4}) is bounded on $\ell_w^2(Q,\R^d)$  for all  $\rho>0$. Furthermore one has the inequality
\be \label{D3}
 \|T\|_w \  \le \ \sup_{\rho>0} \|T_\rho\|_w \ .
\ee
Similarly one can define for a periodic weight $W:Q\times Q\ra\R$ the weighted space  $L_W^2(Q\times Q\times \Om_Q,\R^d)$  as the space of all periodic measurable functions $\Phi:Q\times Q\times \Om_Q\ra\R^d$ with finite norm  $\|\Phi\|_W$ given by
\be \label{F3}
\|\Phi\|_W^2 \ = \ \sum_{(y,z)\in Q\times Q} W(y,z)\langle \ |\Phi(y,z,\phi)|^2\  \rangle_{\Om_Q,m,x,\mu} \  .
\ee
The operator $T_1$ defined by (\ref{AC2}) is bounded on $L_W^2(Q\times Q\times\Om_Q,\R^d\times\R^d)$ since the operator  $T_{1,\rho}$ of (\ref{U4})  is bounded on $\ell_W^2(Q\times Q,\R^d\times\R^d)$  for all  $\rho>0$ and
\be \label{I3}
 \|T_1\|_W \  \le \ \sup_{\rho>0} \|T_{1,\rho}\|_W \ . 
\ee
Finally we define the  weighted space  $L_W^2(Q\times Q\times \Om_Q,\R^d\times\R^d)$ with norm as in (\ref{F3}). Then by the spectral decomposition theorem  the operator $T\otimes I$  is bounded on $L_W^2(Q\times Q\times\Om_Q,\R^d\times\R^d)$ since $T_\rho\otimes I$ is bounded on $l_W^2(Q\times Q,\R^d\times\R^d)$  for all  $\rho>0$. In that case one has the inequality
\be \label{G3}
 \|T\otimes I\|_W \  \le \ \sup_{\rho>0} \|T_\rho\otimes I\|_W \ .
\ee

We can now state a weighted version of Theorem 2.1.   For $\alpha,\beta\in\R$ let $w:Q\ra \R$, $W:Q\times Q\ra \R$  be the weights  $$w(y)=[1+|y|]^\alpha, \ {\rm and} \,\, W(y,z)= [1+|y|]^\alpha[1+\ga(y,z)]^\beta, \ y,z\in Q,$$ where $\ga(y,z)$ is the distance from $y$ to $z$ in the periodic cube $Q$.

\noindent  \textbf{Remark:}   For the proof of Theorem 1.1 we will choose $\beta < -\alpha$ with  $\alpha >0$ small. The identity 
$$W(y,z)\, 1/w(y) \,\delta(y-z) \equiv 1$$ will be used in 
(\ref{AH2}) to obtain a weighted Schwarz inequality which is uniformly bounded in $|x|$.
\begin{theorem} Suppose $Q$ is a periodic cube in $\Z^d$ for some $d\ge 1$ and that $h_j:Q\ra\R^d, \ j=1,2,3$, have the property that  $h_3\in \ell^2_{w}(Q,\R^d)$ and  both $h_1\otimes h_2, \ h_2\otimes h_1$ are in $ \ell^2_{W}(Q\times Q,\R^d\times\R^d)$Then for $|\alpha|, \ |\beta|$ sufficiently small depending only on $\la/\La>1/2$,  there is a positive constant $C(\la,\La)$ depending only on $\la,\La$, such that for any   $ x\in Q, \  \mu\in\R$,
\begin{multline} \label{J3}
|\langle  \ \prod_{j=1}^3\big[ (h_j,\na\phi)-\langle (h_j,\na\phi) \rangle_{\Om_Q,x,m,\mu}\big] \ \rangle_{\Om_Q,x,m,\mu}| \ \le \\
C(\la,\La)\left[ \ \|h_1\otimes h_2\|_{W}+\|h_2\otimes h_1\|_{W} \ \right]\|h_3\|_{w} \sup_{\xi\in \R^d} |V'''(\xi)| \ .
\end{multline}
\end{theorem}
\begin{proof} We first consider the function $\Phi_3(y,\phi)$ defined by (\ref{N2}), whose gradient  $\na\Phi_3(y,\phi)$is given by the Neumann series for the solution of (\ref{AG2}). In view of (\ref{D3}) and Proposition 3.1, the series converges in $L^2_{1/w}(Q\times\Om_Q,\R^d)$ provided $|\alpha|$ is sufficiently small, depending only on $\la/\La>0$, and $\|\na\Phi_3(\cdot,\phi)\|_{1/w}\le C(\la,\La) \|h_3\|_{1/w}$.

Next we consider the function $\Phi:Q\times Q\ra\R^d\times\R^d$ defined by $\Phi(y,z,\phi)=\na\Phi_1(z,\phi)h_2(y)+\na\Phi_2(z,\phi)h_1(y),  \ y,z\in Q$, where the $\Phi_j(\cdot,\phi), \ j=1,2$ are solutions of (\ref{N2}). It follows from (\ref{G3}) and Proposition 3.2 that $\Phi$ is in 
$ L^2_{W}(Q\times Q\times\Om_Q,\R^d\times\R^d)$ if $|\alpha|+|\beta|$  is sufficiently small, depending only on $\la/\La>0$, and $\|\Phi\|_{W}\le C(\la,\La)\left[ \ \|h_1\otimes h_2\|_{W}+\|h_2\otimes h_1\|_{W} \ \right]$. For $\Phi\in L^2_{W}(Q\times Q\times\Om_Q,\R^d\times\R^d)$ we can generate the solution to (\ref{Q2}) by means of the perturbation expansion generated by (\ref{AE2}). It follows then from (\ref{I3}) and Proposition 3.3 that $\na_y\na_z\Psi(y,z,\phi(\cdot)), \ y,z,\in Q$, is in $ L^2_{W}(Q\times Q\times\Om_Q,\R^d\times\R^d)$ if $|\alpha|+|\beta|$  is sufficiently small, depending only on $\la/\La$ with $1/2<\la/\La\le 1$, and $\|\na\na\Psi\|_{W}\le C(\la,\La) \|\Phi\|_{W}$.

To complete the proof of (\ref{J3}) we use the representation (\ref{AH2}). Using the Schwarz inequality as in (\ref{R2})   and $W(y,z)1/w(y)\del(y-z)\equiv 1$  we conclude that 
 \begin{multline} \label{K3}
|\langle  G_1 G_2 G_3    \rangle_{\Om_Q,m,x,\mu}| \ \le  \\
\ \|\na\na\Psi\|_{W} \ |\na\Phi_3\|_{1/w} \ \sup_{\xi\in\R^d}|V'''(\xi)|  \\
\le \ C(\la,\La)\left[ \ \|h_1\otimes h_2\|_{W}+\|h_2\otimes h_1\|_{W} \ \right]\|h_3\|_{1/w} \sup_{\xi\in \R^d} |V'''(\xi)| \ .
\end{multline} 
 \end{proof}

\begin{proof}[Proof of Theorem 1.1]  From \cite{fs} it will be sufficient to obtain an  estimate for
$|\langle \  \big[X-\langle X \rangle_{\Om_Q,x,m,\mu} \big]^3 \ \rangle_{\Om_Q,x,m,\mu}|$ with $X=\phi(0)-\phi(x)$, which is uniform as $Q\ra\Z^d$ and $m\ra 0$. 
 As in (\ref{U1}) we may write $\phi(0)-\phi(x)$ in terms of the gradient of the periodic Greens function, $$h_Q(y)=\lim_{\rho\ra 0} \na G_{\rho,Q}(y)$$
\be \label{M3}
\phi(0)-\phi(x) \ = \ (\na G_{\rho,Q},\na\phi)-(\tau_x \na G_{\rho,Q},\na\phi) +( \ \rho[(G_{\rho,Q}-\tau_x G_{\rho,Q}], \phi \ ) \ ,
\ee
where $\tau_x$ denotes translate of a function by $x$.   We do not use any cancellation between $h_Q$ and its translates.  From the first inequalities of (\ref{D4}), (\ref{E4}) it follows that $\lim_{\rho\ra 0}( \ \rho[(G_{\rho,Q}-\tau_x G_{\rho,Q}], \phi \ )=0$, whence  (\ref{M3}) imply that
\be \label{N3}
\phi(0)-\phi(x) \ = \ ( h_{Q},\na\phi)-(\tau_x  h_{Q},\na\phi) \ .
\ee
In order to prove Theorem 1.1 it will therefore be sufficient for us to apply Theorem 4.1 for $h_1=h_2=h_Q$ and $h_3=h_Q$ or $h_3=\tau_x h_Q$. One easily sees that for $d=2$ and $0<\alpha<d$, there is a constant $C_\alpha$ depending only on $\alpha$  such that
\be \label{O3}
\|\tau_x h_Q\|_{1/w} \ \le \  C_\alpha/[1+|x|]^\alpha \ , \quad x\in Q.
\ee
Similarly one has that for  $d=2$ and $0<\alpha<d, \ -d<\beta<-\alpha$, there is a constant $C_{\alpha,\beta}$ depending only on $\alpha, \beta$  such that
\be \label{P3}
\|h_Q\otimes h_Q\|_{W} \ \le \  C_{\alpha,\beta}\ .
\ee
The inequality  (\ref{M1}) follows from (\ref{O3}), (\ref{P3}) and Theorem 4.1. 
\end{proof}
\begin{proof}[Proof of Theorem 1.2]
We observe that by translation invariance of the measure we only need to take $h_1=h_2=h_Q, \ h_3=\tau_x h_Q$ in Theorem 4.1.
\end{proof}

\thanks{ {\bf Acknowledgement:} J. Conlon is grateful to N. Pattakos and A. Volberg for discussions regarding weighted norm inequalities. T. Spencer is indebted to D. Brydges, J. Fr\"{o}hlich and I.M. Sigal for many helpful conversations. Thanks also to E. Basor, H. Pinson and H. Widom for discussions many years ago concerning dimers.  Finally the authors thank the referees for their detailed comments which have helped to improve the exposition.

\end{document}